\documentclass[11pt,a4paper,reqno]{article}
\usepackage{mysty}
\usepackage{authblk}

\usepackage{fullpage}
\usepackage{subfig}

\usepackage{amsfonts,amssymb}
\usepackage{amsmath}
\usepackage{graphicx}
\usepackage{cite}
\usepackage{enumerate}

\theoremstyle{plain}
\newtheorem{thm}{\protect\theoremname}
\usepackage{babel}
\providecommand{\theoremname}{Theorem}

\SetLabelAlign{parright}{\parbox[t]{\labelwidth}{\raggedleft{#1}}}
\setlist[description]{style=multiline,topsep=4pt,align=parright}

\makeatletter
\let\reftagform@=\tagform@
\def\tagform@#1{\maketag@@@{(\ignorespaces\textcolor{black}{#1}\unskip\@@italiccorr)}}
\newcommand{\iref}[1]{\textup{\reftagform@{\tcr{\ref{#1}}}}}
\makeatother

\begin{document}
	
\title{Design Fast Algorithms For Hodgkin-Huxley Neuronal Networks}
\author{Zhong-Qi Kyle Tian and Douglas Zhou\textsuperscript{\footnote{zdz@sjtu.edu.cn}}}
\affil{School of Mathematical Sciences, MOE-LSC, and Institute of Natural Sciences, Shanghai Jiao Tong University, Shanghai, China}

\date{}
\maketitle

\begin{abstract}
	The stiffness of the Hodgkin-Huxley (HH) equations during an action
	potential (spike) limits the use of large time steps. We observe that
	the neurons can be evolved independently between spikes, $i.e.,$
	different neurons can be evolved with different methods and different
	time steps. This observation motivates us to design fast algorithms
	to raise efficiency. We present an adaptive method, an exponential
	time differencing (ETD) method and a library-based method to deal
	with the stiff period. All the methods can use time steps one order
	of magnitude larger than the regular Runge-Kutta methods to raise
	efficiency while achieving precise statistical properties of the original
	HH neurons like the largest Lyapunov exponent and mean firing rate.
	We point out that the ETD and library methods can stably achieve maximum
	8 and 10 times of speedup, respectively. 
	
\end{abstract}

\section{Introduction}

The Hodgkin-Huxley (HH) system \cite{hodgkin1952quantitative,hassard1978bifurcation,dayan2003theoretical}
is widely used to simulate neuronal networks in computational neuroscience
\cite{hodgkin1952quantitative,hassard1978bifurcation,dayan2003theoretical}.
One of its attractive properties is that it can describe the detailed
generation of action potentials realistically, $e.g.$, the spiking
of the squid's giant axon. For numerical simulation in practice, it
is generally evolved with simple explicit methods, $\mathit{e.g.}$,
the Runge-Kutta scheme, and often with a fixed time step. However,
the HH equations become stiff when the HH neuron fires a spike and
we have to take a sufficiently small time step to avoid stability
problems. The stiffness of HH equations limits its application and as a substitution, the
conductance-based integrate-and-fire (I\&F) systems are often used
in the consideration of efficiency \cite{somers1995emergent,mclaughlin2000neuronal,cai2005architectural,rangan2007fast}
although detailed generation of action potentials in I\&F equations
are omitted. Therefore, it is meaningful to solve the stiff problem
in HH system and design fast algorithms allowing a large time step.

The stiff problem of HH system only occurs during the action potentials
which is from the activities of sodium and potassium ions channels.
Then a natural idea is to use the adaptive method so that we may use
a larger time step outside the stiff period. It works well for a single
HH neuron but fails in large-scaled HH networks \cite{borgers2013exponential}
since there are firing events almost everywhere and the obtained time
step in standard adaptive method is very small. One key point
we observed is that the neurons interact only at the spike moments,
so they can be evolved independently between spikes. Our strategy
for the adaptive method is that we use a large time step to evolve
the network and split it up into subintervals (small time step) to
evolve the neurons that are during the stiff period. 

The exponential time differencing (ETD) method \cite{petropoulos1997analysis,cox2002exponential,borgers2013exponential,ju2018energy},
proposed for stiff systems, may be another effective and efficiency method
for HH equations. The idea of ETD method is that, in each time step,
the equations can be written into a dominant linear term and a relatively
small residual term. By multiplying an integrating
factor, the evolving of HH equations turns to estimate the integral
of the residual, which allows a large time step without causing stability
problem. The recent study \cite{borgers2013exponential} gives a second-order
ETD method for a single HH neuron.
However, for networks, an efficient and high accurate ETD method has
to consider the influence of spike-spike interactions. When evolving
a network for one single time step, there may be several neurons firing
in the time step. Unknowing when and which neurons will fire, we have
to wait until the end of the evolution to consider the spike-induced
effects \cite{hansel1998numerical,shelley2001efficient,borgers2013exponential},
$e.g.$, the change of conductance and membrane potential of postsynaptic
neurons. Without a carefully recalibration, the accuracy of the network
is only the first-order. We point out that the spike-spike correction
procedure introduced in Ref. \cite{rangan2007fast} can solve this
problem by iteratively sorting the possible spike times, updating
the network to the first one and recomputing all future spikes within
the time step. In this Letter, we give a fourth-order ETD method for networks with spike-spike correction.

The above adaptive and ETD methods can achieve high quantitative accuracy
of HH neurons like the spike shapes. But sometimes the goal of numerical
simulations is to achieve qualitative insight or statistical accuracy,
since with uncertainty in the model parameters the HH equations can
only approximately describe the biological reality. In this case,
we offer a library method \cite{sun2009library} to raise efficiency
which can avoid the stiff period by treating the HH neuron as an I\&F
one. The idea of the library method is that once a HH neuron's membrane
potential reaches the threshold, we stop evolving its HH equations
and restart after the stiff period with reset values interpolated
from a pre-computed high resolution data library. Therefore, we avoid
the stiff period and can use a large time step to evolve the HH model.
Compared with the original library method in Ref. \cite{sun2009library},
our library method can significantly simply the way to build the library
and improve the precision of library (see section \ref{sec:Library method}). 

We use the regular fourth-order Runge-Kutta method (RK4) to compare
the performance of the adaptive, ETD and library methods with spike-spike
correction procedure included in all the methods. The three advanced
methods can use time steps one order of magnitude larger than that
in the regular RK4 method while achieving precise statistical properties
of the HH model, $e.g.,$ firing rates and chaotic dynamical property.
We also give a detailed comparison of efficiency for the given methods.
Our numerical results show that the ETD and library methods can achieve
stable high times of speedup with a maximum 8 and 10 times, respectively. 

The outline of the paper is as follows. In Section 2, we give the
equations of the HH model. In Section 3, 4 and 5, we describe the
regular RK4, adaptive and ETD method, respectively. In Section 6,
we give the details of how to build and use the data library. In Section
7, we give the numerical results. We discuss and conclude in Section
8.

\section{The model \label{sec:The-model}}

The dynamics of the $i$th neuron of a Hodgkin-Huxley (HH) network
with $N$ excitatory neurons is governed by 

\begin{equation}
C\frac{dV_{i}}{dt}=-(V_{i}-V_{Na})G_{Na}m_{i}^{3}h_{i}-(V_{i}-V_{K})G_{K}n_{i}^{4}-(V_{i}-V_{L})G_{L}+I_{i}^{\textrm{input}}\label{eq: V of HH}
\end{equation}

\begin{equation}
\frac{dz_{i}}{dt}=(1-z_{i})\alpha_{z}(V_{i})-z_{i}\beta_{z}(V_{i})\text{ for }z=m,h\text{ and }n\label{eq:mhn of HH}
\end{equation}
where $V_{i}$ is the membrane potential, $m_{i}$, $h_{i}$ and $n_{i}$
are gating variables, $V_{Na},V_{K}$ and $V_{L}$ are the reversal
potentials for the sodium, potassium and leak currents, respectively,
$G_{Na},G_{K}$ and $G_{L}$ are the corresponding maximum conductances.
$I_{i}^{\textrm{input}}=-G_{i}(t)(V_{i}-V_{G})$ is the input current
with 
\begin{equation}
\frac{dG_{i}(t)}{dt}=-\frac{G_{i}(t)}{\sigma_{r}}+H_{i}(t)
\end{equation}

\begin{equation}
\frac{dH_{i}(t)}{dt}=-\frac{H_{i}(t)}{\sigma_{d}}+f\sum_{l}\delta(t-s_{il})+\sum_{j\neq i}\sum_{l}S_{ij}\delta(t-\tau_{jl})\label{eq:f input}
\end{equation}
where $V_{G}$ is the reversal potential, $G_{i}(t)$ is the conductance,
$H_{i}(t)$ is an additional parameter to smooth $G_{i}(t)$, $\sigma_{r}$
and $\sigma_{d}$ are fast rise and slow decay time scale, respectively,
and $\delta(\cdot)$ is the Dirac delta function. The second term
in Eq. (\ref{eq:f input}) is the feedforward input with magnitude
$f$. The input time $s_{il}$ is generated from a Poisson process
with rate $\nu$. The third term in Eq. (\ref{eq:f input}) is the
synaptic current from synaptic interactions in the network, where
$S_{ij}$ is the coupling strength from the $j$th neuron to the $i$th
neuron, $\tau_{jl}$ is the $l$th spike time of $j$th neuron. The
forms of $\alpha$ and $\beta$ and other model parameters are given
in Appendix.

When the voltage $V_{i}$, evolving continuously according to Eqs.
(\ref{eq: V of HH}, \ref{eq:mhn of HH}), reaches the threshold $V^{\textrm{th}}$,
we say the $i$th neuron fires a spike at this time. Instantaneously,
all its postsynaptic neurons receive this spike and their corresponding
parameter $H$ jumps by an appropriate amount $S_{ji}$ for the $j$th
neuron. For the sake of simplicity, we mainly consider a homogeneously
and randomly connected network with $S_{ij}=A_{ij}S$ where $\mathbf{A}=(A_{ij})$
is the adjacency matrix and $S$ is the coupling strength. But note
that the conclusions shown in this paper will not change if we extend
the given methods to more complicated networks, $e.g.,$ networks
of both excitatory and inhibitory neurons, more realistic connectivity
with coupling strength following the typically Log-normal distribution
\cite{song2005highly,ikegaya2012interpyramid}. 

\section{Regular method}

We first introduce the regular Runge-Kutta fourth-order scheme (RK4)
with fixed time step $\Delta$t to evolve the HH model. For the easy
of illustration, we use the vector 
\begin{equation}
X_{i}(t)=(V_{i}(t),m_{i}(t),h_{i}(t),n_{i}(t),G_{i}(t),H_{i}(t))\label{eq:stats X}
\end{equation}
to represent the variables of the $i$th neuron. The neurons interact
with each other through the spikes by changing the postsynaptic state
of $X$, so it is important to obtain accurate spike sequences, $e.g.,$
at least with an accuracy of fourth-order. We determine the spike
times as follows \cite{hansel1998numerical,shelley2001efficient}.
If neuron $i$ fires in $[t_{n},t_{n+1}]$, $i.e.,$ $V_{i}(t_{n})<V^{\text{th}}$
and $V_{i}(t_{n+1})\geq V^{\text{th}}$, we can use the membrane potential
and its derivative at $t=t_{n}$ and $t=t_{n+1}$: $V_{i}(t_{n}),\frac{dV_{i}}{dt}(t_{n}),V_{i}(t_{n+1}),\frac{dV_{i}}{dt}(t_{n+1})$
to perform a cubic Hermite interpolation $p(t)$ to decide the spike
time by finding the solution $t_{\text{spike}}$ of $p(t)=V^{\text{th}}$
with $t_{n}<t_{\text{spike}}$$\leq t_{n+1}$. Then the obtained spike
time $t_{\text{spike}}$ has an accuracy of fourth-order given that
the membrane potential and its derivative has an accuracy of fourth-order
as well. 

In between spikes, the RK4 method can achieve a fourth-order of accuracy.
However, for a spike-containing time step, it requires a careful recalibration
of $X$ to account for the spike-spike interaction. For example, consider
2 bidirectionally connected neurons $1$ and $2$. If neuron 1 fires
in the time step, then it is necessary to recalibrate $X_{2}$ to
achieve an accuracy of fourth-order. If neuron 2 also fires, then
$X_{1}$ and spike time of neuron 1 is imprecise and we should recalibrate
them. Then again $X_{2}$ and the spike time of neuron 2 is imprecise. 

\begin{figure}[H]
	\begin{centering}
		\includegraphics[width=0.5\textwidth]{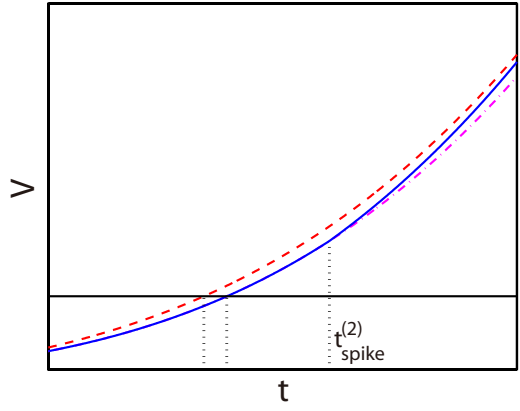}
		\par\end{centering}
	\caption{Local error analysis of neuron $1$ for $t_{0}\leq t\leq t_{0}+\Delta t$
		. The blue solid line is true membrane potential $v_{1}$, magenta
		dashdot line is true membrane potential $\tilde{v}_{1}$ without considering
		the spike from neuron 2 and the red dashed line is the cubic Hermite
		interpolation $p_{1}$. The solid horizontal line is the threshold
		and the dotted vertical lines are spike times. \label{fig:illustrate SSC} }
\end{figure}

We take the spike-spike correction procedure \cite{rangan2007fast}
to solve this problem. The strategy is that we preliminarily evolve
neuron 1 and 2 considering only the feedforward input in the time
step from $t_{0}$ to $t_{0}+\Delta t$ and decide the spike times
$t_{\text{spike}}^{(1)}$ and $t_{\text{spike}}^{(2)}$ by a cubic
Hermite interpolation. Suppose that $t_{\text{spike}}^{(1)}<t_{\text{spike}}^{(2)}$.
Let $v_{1}(t)$ denote the true membrane potential of neuron 1, $\tilde{v}_{1}(t)$
denote the true membrane potential of neuron 1 without considering
the spike of neuron 2 fired at $t_{\text{spike}}^{(2)}$ and $p_{1}(t)$
be the cubit Hermite interpolation for neuron 1 as shown in Fig. \ref{fig:illustrate SSC}.
Then we have 
\begin{equation}
p_{1}(t)-\tilde{v}_{1}(t)=O(\Delta t^{4})\text{ for }t_{0}\leq t\leq t_{0}+\Delta t.
\end{equation}
Note that the spike from neuron 2 starts to affect at time $t_{\text{spike}}^{(2)}$,
$i.e.$, 
\begin{equation}
v_{1}(t)=\tilde{v}_{1}(t)\text{ for }t_{0}\leq t\leq t_{\text{spike}}^{(2)}.
\end{equation}
Therefore, the preliminarily computed $t_{\text{spike}}^{(1)}$ indeed
has an accuracy of four-order as shown in Fig. \ref{fig:illustrate SSC}.
Then we can update neuron 1 and 2 from $t_{0}$ to $t_{\text{spike}}^{(1)}$,
accept the spike time of neuron 1 at $t_{\text{spike}}^{(1)}$, and
evolve them to $t_{0}+\Delta t$ to obtain spike time of neuron 2,
$X_{1}(t_{0}+\Delta t)$ and $X_{2}(t_{0}+\Delta t)$ with an accuracy
of fourth-order. For large networks, the strategy is the same and
detailed algorithm of regular RK4 scheme is given in Algorithm 1.

\begin{algorithm}[H] 	  
	\caption{Regular RK4 algorithm} 
	\KwIn{An initial time $t_0$, time step $\Delta t$, feedforward input times $\{s_{il}\}$} 	
	\KwOut{$\{X_i(t_0+\Delta t)\}$ and $\{\tau_{il}\}$(if any fired)} 
	Preliminarily evolve the network from $t_0$ to $t_0+\Delta t$ to find the first synaptic spike: 
	
	\For{$i = 1$ to $N$} 
	{	 		 	
		Let $M$ denote the total number of feedforward spikes of the $i$th neuron within $[t_0, t_0+\Delta t]$ and sort them into an increasing list $\{T_m^{\text{sorted}}\}$. Then we extend this notation such that 		$T_0^{\text{sorted}} := t_0$ and $T_{M+1}^{\text{sorted}} := t_0+\Delta t$.\\	 			
		\For{$m = 1$ to $M+1$ } 		
		{ 		
			Advance the equations for the $i$th HH neuron from  $T_{m-1}^{\text{sorted}}$ to $T_m^{\text{sorted}}$ using the standard RK4 scheme. Then update the conductance $H_i(T_m^{\text{sorted}})$ by adding $f$.\\			
			
			\If{$V_i(T_{m-1}^{\text{sorted}}) < V^{\textrm{th}}$, $V_i(T_m^{\text{sorted}}) \geq V^{\textrm{th}}$}		
			{ 	
				We know that the $i$th neuron spiked during $[T_{m-1}^{\text{sorted}}, T_{m}^{\text{sorted}}]$.\\						
				Find the spike time $t_{\text{spike}}^{(i)}$ by a cubic Hermite interpolation using the values $V_i(T_{m-1}^{\text{sorted}})$, $V_i(T_m^{\text{sorted}})$, 		$\frac{dV_i}{dt}(T_{m-1}^{\text{sorted}})$ and $\frac{dV_i}{dt}(T_m^{\text{sorted}})$.\\ 	 		
			} 	
			
		} 
	}
	
	\While{$\{t_{\text{spike}}^{(i)}\}$ is non-empty}	 
	{
		Find the minimum of $\{t_{\text{spike}}^{(i)}\}$, say $t_{\text{spike}}^{(i_1)}$.			 	
		Evolve neuron $i_1$  to $t_{\text{spike}}^{(i_1)}$, generate a spike at this moment and update its postsynaptic neurons to $t_{\text{spike}}^{(i_1)}$. Preliminarily evolve the affected neurons from $t_{\text{spike}}^{(i_1)}$ to $t_0+\Delta t$ to update their spike times  $\{t_{\text{spike}}\}$.	
	}							 					 
	{ 		
		We accept $\{X_i{(T_{M+1}^{\text{sorted}})}\}$ as the solution  $\{X_i{(t+\Delta t)}\}$;\\ 	
	} 
\end{algorithm}

\section{Adaptive method}

When a neuron fires a spike, the HH neuron equations are stiff for
some milliseconds, denoted by $T^{\textrm{stiff}}$ as shown in Fig.
\ref{fig:Typical-fire-pattern}. This stiff period requires a sufficiently
small time step to avoid stability problem. In the regular RK4 scheme,
we use fixed time step, so to satisfy the requirement of stability,
we have to use a relatively small time step, $\mathit{e.g.}$, $\Delta t=1/32$
ms. But we may need to simulate the HH model frequently to analyze
the system's behavior with different parameters or evolve the model
for a long run time (hours) to obtain precise statistical properties.
Therefore, it is important to enlarge the time step as much as possible
to raise efficiency. 

\begin{figure}[h]
	\begin{centering}
		\includegraphics[width=1\textwidth]{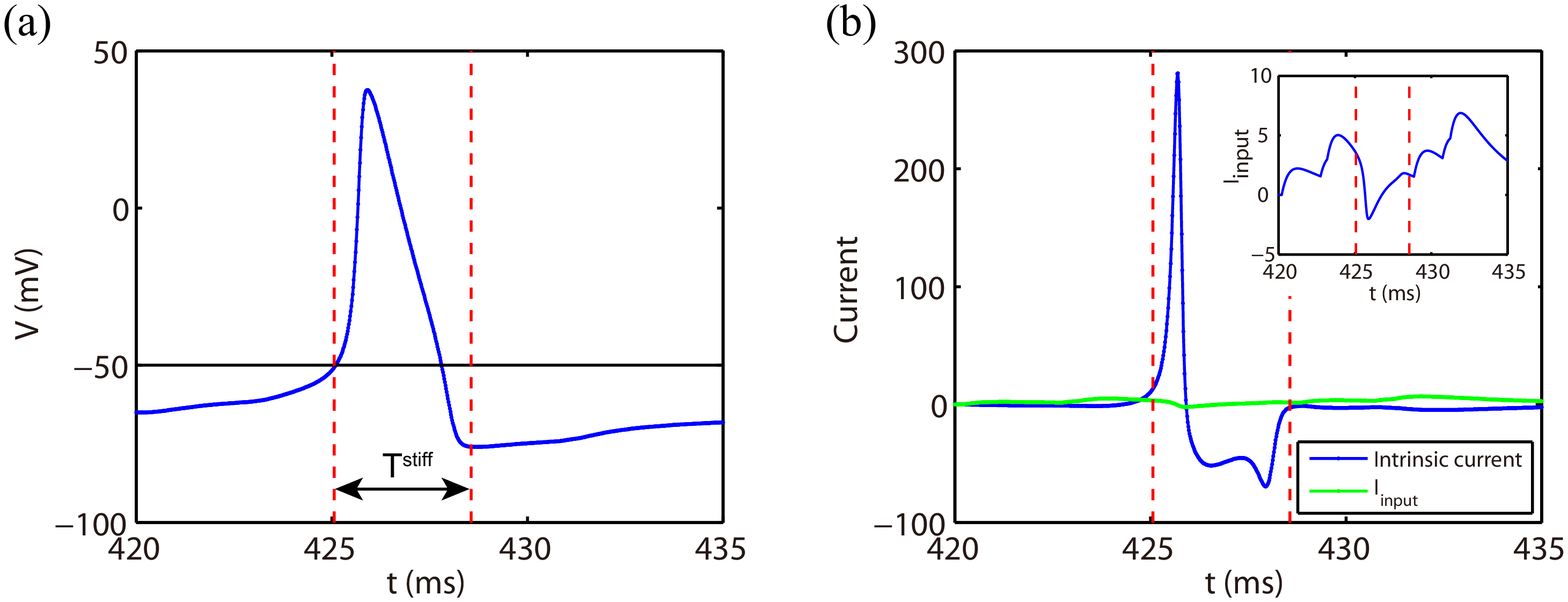}
		\par\end{centering}
	\caption{Typical firing event of a single HH neuron. (a) The trajectory of
		voltage $V$. The black line indicates the threshold $V^{\textrm{th}}$
		and the red dotted lines indicate the stiff period. (b) The trajectory
		of the intrinsic current and input current $I_{i}^{\textrm{input}}$
		($\mu\textrm{A\ensuremath{\cdot}cm}^{-2}$). The intrinsic current
		is the sum of ionic and leakage currents. \label{fig:Typical-fire-pattern}}
\end{figure}

We note that between spikes the neurons do not affect each other and
thus can be evolved independently. As shown in Algorithm 1, we induce
the spike-spike correction to obtain a high accurate method of fourth-order,
$i.e.$, we should split up the time step once a presynaptic neuron
fires. Therefore, the neurons are indeed evolved independently in
simulation. With this observation, we can design efficient method
by reasonably treating the neurons inside the stiff period.

We first introduce our adaptive method. Note that the derivative of
voltage is quite small outside the stiff period as shown in Fig. \ref{fig:Typical-fire-pattern}(b).
Then strategy is that we take a large time step $\Delta t$ to evolve
the network. Once a neuron is in the stiff period, the large time
step is then split up into small time steps $\Delta t_{S}$, $e.g.,$
with a value of $\Delta t_{S}=1/32$ ms in this Letter, as shown in
Fig. \ref{fig:Illustration-adaptive}. For each time step, we still
use the standard RK4 scheme to evolve the neurons, so the adaptive
method has an accuracy of $O(\Delta t^{4})$. Detailed adaptive algorithm
is the same as the Algorithm 1 except that the step 3 in Algorithm
1 should be replaced by the following algorithm. 

\begin{figure}[h]
	\begin{centering}
		\includegraphics[width=0.5\textwidth]{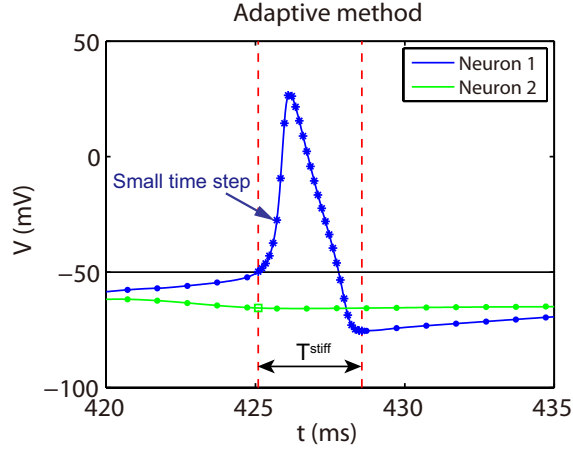}
		\par\end{centering}
	\caption{Illustration of the adaptive method for two HH neurons unidirectionally
		connected from neuron 1 to neuron 2. After neuron 1 fires a spike,
		it is evolved with a small time step $\Delta t_{S}$ during the stiff
		period indicated by the stars. The dots indicate the time nodes for
		the large time step. The green square indicates an update time node
		due to the synaptic spike from neuron 1. \label{fig:Illustration-adaptive}}
	
\end{figure}

\begin{algorithm}[H] 	
	\caption{Adaptive algorithm}	
	\eIf{The $i$th neuron is outside the stiff period}
	{
		Let $M$ denote the total number of feedforward spikes of the $i$th neuron within $[t_0, t_0+\Delta t]$.
	}
	{
		Let $M$ denote the total number of feedforward spikes of the $i$th neuron within $[t_0, t_0+\Delta t]$ and the extra time nodes $t_0+\Delta t_{S},t_0+2\Delta t_{S},...,t_0+k\Delta t_{S}$,
		$k=[\frac{\Delta t}{\Delta t_{S}}]$ ($[\cdot]$ takes the round-off
		number).  
	}	
	Sort them into an increasing list $\{T_m^{\text{sorted}}\}$ and extend this notation such that $T_0^{\text{sorted}} := t_0$ and $T_{M+1}^{\text{sorted}} := t_0+\Delta t$.\\		
\end{algorithm}

\section{Exponential time differencing method}

We now introduce the exponential time differencing method (ETD) which
is proposed for stiff systems \cite{cox2002exponential,borgers2013exponential,petropoulos1997analysis}.
To describe the ETD methods for HH networks, it is more instructive
to first consider a simple ordinary differential equation

\begin{equation}
\frac{du}{dt}=F(t,u)
\end{equation}
where $F(t,u)$ represents the stiff nonlinear forcing term. For a
single time step from $t=t_{k}$ to $t=t_{k+1}=t_{k}+\Delta t$, we
rewrite the equation 
\begin{equation}
\frac{du}{dt}=au+b+(F(t,u)-au-b)\label{eq:ETD_2}
\end{equation}
where $au+b$ is the dominant linear approximation of $F(t,u)$, and
$F(t,u)-au-b$ is the residual error. By multiplying Eq. (\ref{eq:ETD_2})
through by the integrating factor $e^{-at}$, we can obtain

\begin{equation}
u(t_{k+1})=u(t_{k})e^{a\Delta t}+e^{a\Delta t}\int_{0}^{\Delta t}e^{-a\tau}\tilde{F}(t_{k}+\tau,u(t_{k}+\tau))d\tau.\label{eq:ETD_integral}
\end{equation}
where $\tilde{F}(t,u)=F(t,u)-au$. This formula is exact, and the
essence of the ETD method is to find a proper way to approximate the
integral in this expression, $e.g.,$ approximating $\tilde{F}$ by
a polynomial. 

Here we give an ETD method with Runge-Kutta fourth-order time stepping
(ETD4RK). We write $u_{k}$ and $\tilde{F}_{k}$ for the numerical
approximations for $u(t_{k})$ and $\tilde{F}(t_{k},u_{k})$, respectively
(We use similar writing in the derivation of ETD4RK method for HH
equations). The ETD4RK method is given by

\noindent 
\begin{align}
a_{k} & =u_{k}e^{a\Delta t/2}+\tilde{F}_{k}\frac{e^{a\Delta t/2}-1}{a}\\
b_{k} & =u_{k}e^{a\Delta t/2}+\tilde{F}(t_{k}+\Delta t/2,a_{k})\frac{e^{a\Delta t/2}-1}{a}\\
c_{k} & =a_{k}e^{a\Delta t/2}+(2\tilde{F}(t_{k}+\Delta t/2,b_{k})-\tilde{F}_{k})\frac{e^{a\Delta t/2}-1}{a}\\
u_{k+1} & =u_{k}e^{a\Delta t}+g_{0}\tilde{F}_{k}+g_{1}[\tilde{F}(t_{k}+\Delta t/2,a_{k})+\tilde{F}(t_{k}+\Delta t/2,b_{k})]+g_{2}\tilde{F}(t_{k}+\Delta t,c_{k})
\end{align}
where 
\begin{alignat}{1}
g_{0} & =[-4-a\Delta t+e^{a\Delta t}(4-3a\Delta t+a^{2}\Delta t^{2})]/a^{3}\Delta t^{2}\nonumber \\
g_{1} & =2[2+a\Delta t+e^{a\Delta t}(-2+a\Delta t)]/a^{3}\Delta t^{2}\label{eq:g}\\
g_{2} & =[-4-3a\Delta t-a^{2}\Delta t^{2}+e^{a\Delta t}(4-a\Delta t)]/a^{3}\Delta t^{2}\nonumber 
\end{alignat}

The above formula requires an appropriate linear approximation $au+b$
for $F(t,u)$ during the single time step, so that the residual error
is relatively small and we can use a large time step to raise efficiency.
Note that the variables $V,m,h,$ and $n$ appear linearly in the
HH equations (\ref{eq: V of HH}, \ref{eq:mhn of HH}). Therefore,
we can directly apply ETD4RK method to HH system. For the easy of
writing, we give the standard ETD4RK method for a single HH neuron
to evolve over a single time step from $t=t_{k}$ to $t=t_{k+1}=t_{k}+\Delta t$.
The HH equations (\ref{eq: V of HH}, \ref{eq:mhn of HH}) are rewritten
in the form 
\begin{equation}
\frac{dz}{dt}=a_{z}z+\tilde{F}_{z}(t,V,m,h,n)\text{ for }z=V,m,h\text{ and }n
\end{equation}
where 

\begin{equation}
a_{V}=(-G_{Na}m_{k}^{3}h_{k}-G_{K}n_{k}^{4}-G_{L})/C
\end{equation}

\begin{equation}
a_{z}=-\alpha_{z}(V_{k})-\beta_{z}(V_{k})\text{ for }z=m,h\text{ and }n
\end{equation}
and
\begin{equation}
\tilde{F}_{V}=[-(V-V_{Na})G_{Na}m^{3}h-(V-V_{K})G_{K}n^{4}-(V-V_{L})G_{L}+I^{\textrm{input}}-a_{V}V]/C
\end{equation}

\begin{equation}
\tilde{F}_{z}=(1-z)\alpha_{z}(V)-z\beta_{z}(V)-a_{z}z\text{ for }z=m,h\text{ and }n.
\end{equation}
The standard ETD4RK method is given by

\begin{align}
a_{z,k} & =z_{k}e^{a_{z}\Delta t/2}+\tilde{F}_{z,k}\frac{e^{a_{z}\Delta t/2}-1}{a_{z}}\\
b_{z,k} & =z_{k}e^{a_{z}\Delta t/2}+\tilde{F_{z}}(t_{k}+\Delta t/2,a_{V,k},a_{m,k},a_{h,k},a_{n,k})\frac{e^{a_{z}\Delta t/2}-1}{a_{z}}\\
c_{z,k} & =a_{z,k}e^{a_{z}\Delta t/2}+(2\tilde{F}_{z}(t_{k}+\Delta t/2,b_{V,k},b_{m,k},b_{h,k},b_{n,k})-\tilde{F}_{z,k})\frac{e^{a_{z}\Delta t/2}-1}{a_{z}}
\end{align}
\begin{equation}
\begin{aligned}z_{k+1} & =z_{k}e^{a_{z}\Delta t}+g_{z,0}\tilde{F}_{z,k}+g_{z,1}\tilde{F_{z}}(t_{k}+\Delta t/2,a_{V,k},a_{m,k},a_{h,k},a_{n,k})\\
& +g_{z,1}\tilde{F}_{z}(t_{k}+\Delta t/2,b_{V,k},b_{m,k},b_{h,k},b_{n,k})+g_{z,2}\tilde{F}_{z}(t_{k}+\Delta t,c_{V,k},c_{m,k},c_{h,k},c_{n,k})
\end{aligned}
\end{equation}
for $z=V,m,h$ and $n$. The forms of $g_{z,0},g_{z,1},g_{z,2}$ are
the same as given in Eq. (\ref{eq:g}) except that the term $a$ should
be replaced by $a_{z}$.

\begin{figure}[H]
	\begin{centering}
		\includegraphics[width=0.5\textwidth]{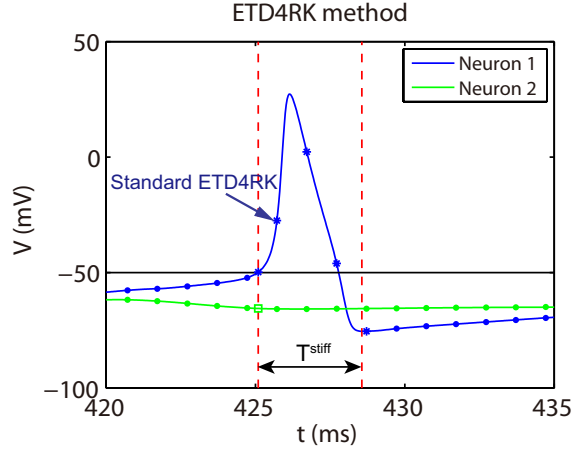}
		\par\end{centering}
	\caption{Illustration of the ETD4RK method for two HH neurons unidirectionally
		connected from neuron 1 to neuron 2. After neuron 1 fires a spike,
		we use the standard ETD4RK scheme to evolve it during the stiff period
		indicated by the stars. The dots indicate the time nodes where we
		use the standard RK4 scheme. The green square indicates an update
		time node due to the synaptic spike from neuron 1. \label{fig:Illustration-ETD4RK}}
	
\end{figure}

Compared with the standard RK4 scheme, the standard ETD4RK scheme
requires extra calculation for the linear approximation terms in Eqs.
(\ref{eq: V of HH}, \ref{eq:mhn of HH}). Besides the HH equations
are stiff only during an action potential as shown in Fig. \ref{fig:Typical-fire-pattern}.
Therefore, it is more suitable to design the ETD4RK method for HH
networks as following: If a neuron is inside the stiff period, we
use the standard ETD4RK scheme to evolve its HH equations, otherwise,
we still use the standard RK4 scheme, as shown in Fig. \ref{fig:Illustration-ETD4RK}.
Detailed ETD4RK algorithm is also based on the Algorithm 1 except
that the step 5 should be replaced by the following algorithm.

\begin{algorithm}[H] 	
	\caption{ETD4RK algorithm}
	\eIf{The $i$th neuron is outside the stiff period}
	{
		Advance the equations for the $i$th HH neuron from  $T_{m-1}^{\text{sorted}}$ to $T_m^{\text{sorted}}$ using the standard RK4 scheme. 
	}
	{
		Advance the equations for the $i$th HH neuron from  $T_{m-1}^{\text{sorted}}$ to $T_m^{\text{sorted}}$ using the standard ETD4RK scheme.  
	}	
	Then update the conductance $H_i(T_m^{\text{sorted}})$ by adding $f$.	
\end{algorithm}

\section{Library method \label{sec:Library method}}

The introduced adaptive and ETD4RK methods can evolve the HH neuronal
networks quite accurately, $e.g.$, they both can capture precise
spike shapes. The library method \cite{sun2009library} sacrifices
the accuracy of the spike shape by treating the HH neuron's firing
event like an I\&F neuron. Once a neuron's membrane potential reaches
the threshold $V^{\textrm{th}}$, we stop evolving its $V,m,h,n$
for the following stiff period $T^{\textrm{stiff}}$, and restart
with the values interpolated from a pre-computed high resolution library.
Thus the library method can achieve the highest efficiency in principle
by skipping the stiff period. Besides, it can still capture accurate
statistical properties of the HH neuronal networks like the mean firing
rates and chaotic dynamics. 

\subsection{Build the library}

We first describe how to build the library in detail. When a neuron's
membrane potential reaches the threshold $V^{\textrm{th}}$, we record
the values $I^{\textrm{input}},m,h,n$ and denote them by $I^{\textrm{th}},m^{\textrm{th}},h^{\textrm{th}},n^{\textrm{th}}$,
respectively. If we know the exact trajectory of $I^{\textrm{input}}$
for the following stiff period $T^{\textrm{stiff}}$, we can use a
sufficiently small time step to evolve the Eqs. (\ref{eq: V of HH},
\ref{eq:mhn of HH}) for $T^{\textrm{stiff}}$ with initial values
$V^{\textrm{th}},m^{\textrm{th}},h^{\textrm{th}},n^{\textrm{th}}$
to obtain high resolution trajectories of $V,m,h,n$. The values at
the end of the stiff period are like the reset values in I\&F neurons
and are denoted by $V^{\textrm{re}},m^{\textrm{re}},h^{\textrm{re}},n^{\textrm{re}}$.

However it is impossible to obtain the exact trajectory of $I^{\textrm{input}}$
without knowing the feedforward and synaptic spike information. As
shown in Fig. \ref{fig:Typical-fire-pattern}(b), $I^{\textrm{input}}$
varies during the stiff period with peak value in the range of $O(5)$
$\mu\textrm{A\ensuremath{\cdot}cm}^{-2}$, while the intrinsic current,
the sum of ionic and leakage current, is about 30 $\mu\textrm{A\ensuremath{\cdot}cm}^{-2}$
at the spike time, and quickly rises to the peak value about 250 $\mu\textrm{A\ensuremath{\cdot}cm}^{-2}$,
then stays at $O(-50)$ $\mu\textrm{A\ensuremath{\cdot}cm}^{-2}$
in the remaining stiff period. Therefore, the intrinsic current is
dominant in the stiff period. With this observation, we take $I^{\textrm{input}}$
as constant throughout the whole stiff period when we build the library.
We emphasize that this is the only assumption made in the library
method. Then, given a suite of $I^{\textrm{th}},m^{\textrm{th}},h^{\textrm{th}},n^{\textrm{th}}$,
we can obtain the corresponding suite of $V^{\textrm{re}},m^{\textrm{re}},h^{\textrm{re}},n^{\textrm{re}}$. 

\begin{figure}[H]
	\begin{centering}
		\includegraphics[width=0.5\textwidth]{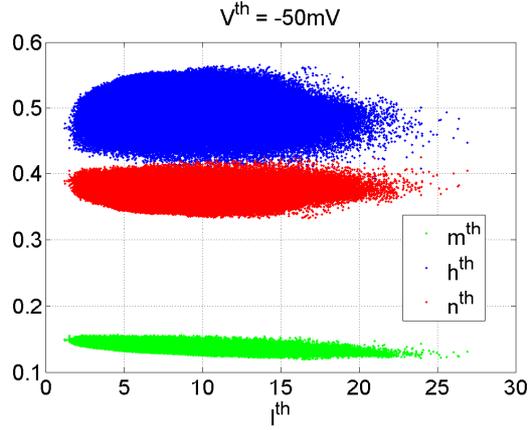}
		\par\end{centering}
	\caption{The ranges of values $I^{\textrm{th}},m^{\textrm{th}},h^{\textrm{th}},n^{\textrm{th}}$.
		\label{fig:Imhn range} }
\end{figure}

When building the library, we choose $N_{I},N_{m},N_{h},N_{n}$ different
values of $I^{\textrm{th}},m^{\textrm{th}},h^{\textrm{th}},n^{\textrm{th}}$,
respectively, equally distributed in their ranges as shown in Fig.
\ref{fig:Imhn range}. For each suite of $I^{\textrm{th}},m^{\textrm{th}},h^{\textrm{th}},n^{\textrm{th}}$,
we evolve the Eqs. (\ref{eq: V of HH}, \ref{eq:mhn of HH}) for a
time interval of $T^{\textrm{stiff}}$ to obtain $V^{\textrm{re}},m^{\textrm{re}},h^{\textrm{re}},n^{\textrm{re}}$
with RK4 method with a sufficiently small time step, $\mathit{e.g.}$,
$\Delta t=2^{-16}$ ms. Note that we should take $I^{\textrm{input}}=I^{\textrm{th}}$
throughout the whole time interval $T^{\textrm{stiff}}$. The library
is built with a total size of $8N_{I}N_{m}N_{h}N_{n}$. In our simulation,
we take the ranges $[0,50]$ $\mu\textrm{A\ensuremath{\cdot}cm}^{-2}$,
$[0,0.3]$, $[0.2,0.6]$ and $[0.3,0.6]$ for $I^{\textrm{th}},m^{\textrm{th}},h^{\textrm{th}},n^{\textrm{th}}$
respectively that can cover almost all possible situations and take
sample numbers $N_{I}=21,N_{m}=16,N_{h}=21,N_{n}=16$ which can make
the library quite accurate. The library occupies only 6.89 megabyte
in binary form and is quite small for today's computers. 

One key point in building the library is to choose a proper threshold
value $V^{\textrm{th}}$. The threshold should be relatively low to
keep the HH equations not stiff and allow a large time step with the
stability requirement satisfied. On the other hand, it should be relatively
high that a neuron will definitely fire a spike after its membrane
potential reaches the threshold. In this Letter, we take $V^{\textrm{th}}=-50$
mV. Correspondingly, we take the stiff period $T^{\textrm{stiff}}=3.5$
ms which is long enough to cover the stiff parts in general firing
events. 

\subsection{Use the library}

\begin{figure}[H]
	\begin{centering}
		\includegraphics[width=0.5\textwidth]{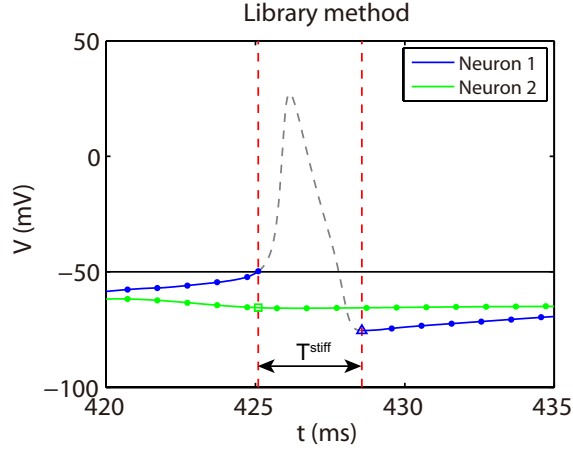}
		\par\end{centering}
	\caption{Illustration of the library method for two HH neurons unidirectionally
		connected from neuron 1 to neuron 2. After neuron 1 fires a spike,
		we compute and record the values $I^{\textrm{th}},m^{\textrm{th}},h^{\textrm{th}},n^{\textrm{th}}$,
		then we stop evolving its $V,m,h,n$ for the following $T^{\textrm{stiff}}$
		ms and restart with the values $V^{\textrm{re}},m^{\textrm{re}},h^{\textrm{re}},n^{\textrm{re}}$
		interpolated from the library indicated by the triangle. The dots
		indicate the time nodes where we use the standard RK4 scheme. The
		green square indicates an update time node due to the synaptic spike
		from neuron 1. \label{fig:Illustration-Lib}}
	
\end{figure}

We now illustrate how to use the library. Once a neuron's membrane
potential reaches the threshold, we first record the values $I^{\textrm{th}},m^{\textrm{th}},h^{\textrm{th}},n^{\textrm{th}}$,
then stop evolving this neuron's equations of $V,m,h,n$ for the following
$T^{\textrm{stiff}}$ ms and restart with values $V^{\textrm{re}},m^{\textrm{re}},h^{\textrm{re}},n^{\textrm{re}}$
linearly interpolated from the pre-computed high resolution data library
as shown in Fig. \ref{fig:Illustration-Lib}. For the easy of writing,
suppose $I^{\textrm{th}}$ falls between two data points $I_{0}^{\textrm{th}}$
and $I_{1}^{\textrm{th}}$ in the library. Simultaneously we can find
the data points $m_{0}^{\textrm{th}}$ and $m_{1}^{\textrm{th}}$,
$h_{0}^{\textrm{th}}$ and $h_{1}^{\textrm{th}}$, $n_{0}^{\textrm{th}}$
and $n_{1}^{\textrm{th}}$, respectively. So we need 16 suites of
values in the library to perform a linear interpolation 

\begin{equation}
z^{\textrm{re}}=\sum_{i,j,k,l=0,1}z^{\textrm{re}}(I_{i}^{\textrm{th}},m_{j}^{\textrm{th}},h_{k}^{\textrm{th}},n_{l}^{\textrm{th}})\frac{I^{\textrm{th}}-I_{1-i}^{\textrm{th}}}{I_{i}^{\textrm{th}}-I_{1-i}^{\textrm{th}}}\frac{m^{\textrm{th}}-m_{1-j}^{\textrm{th}}}{m_{j}^{\textrm{th}}-m_{1-j}^{\textrm{th}}}\frac{h^{\textrm{th}}-h_{1-k}^{\textrm{th}}}{h_{k}^{\textrm{th}}-h_{1-k}^{\textrm{th}}}\frac{n^{\textrm{th}}-n_{1-l}^{\textrm{th}}}{n_{l}^{\textrm{th}}-n_{1-l}^{\textrm{th}}}\label{eq:Lib-V}
\end{equation}
for $z=V,m,h,n$. We note that the parameters $G$ and $H$ have analytic
solutions, so they can be evolved as usual during the stiff period.
After obtaining the high resolution library, we can use a large time
step to evolve the HH neuron network with standard RK4 method. Detailed
library algorithm is the same as the Algorithm 1 except that the step
13 should be replaced by the following algorithm. 

\begin{algorithm}[H] 	
	\caption{Library algorithm}		
	Find the minimum of $\{t_{\text{spike}}^{(i)}\}$, say $t_{\text{spike}}^{(i_1)}$.			 		
	Evolve neuron $i_1$  to $t_{\text{spike}}^{(i_1)}$,  and generate a spike at this moment.\\ 	
	Record the values $I^{\textbf{th}}, m^{\textbf{th}}, h^{\textbf{th}}, n^{\textbf{th}}$, then perform a linear interpolation from the library to get $V^{\textbf{re}}, m^{\textbf{re}}, h^{\textbf{re}}, n^{\textbf{re}}$. Meanwhile, we stop evolving $V, m, h , n$ of neuron $i_1$ for the next $T^{\textbf{stiff}}$ ms, but we still evolve the conductance parameters $G, H$ as usual.\\ 	
	Update its postsynaptic neurons to $t_{\text{spike}}^{(i_1)}$. Preliminarily evolve the affected neurons from $t_{\text{spike}}^{(i_1)}$ to $t_0+\Delta t$ to update their spike times  $\{t_{\text{spike}}\}$.	
\end{algorithm}

\subsection{Error of library method \label{subsec:Error-of-library}}

There are three kinds of error in the library method: the error from
the time step $\Delta t$, the error from the way of interpolation
to use the library and the error from the assumption that we keep
$I^{\textrm{input}}$ as constant throughout the stiff period $T^{\textrm{stiff}}$.
The first one is simply $O(\Delta t^{4})$ since the library method
is based on the RK4 scheme. The other two kinds of error come from
the call of library. For simplicity, we consider the error of voltage
for one single neuron to illustrate
\begin{equation}
\Delta V=|V_{\textrm{library}}-V_{\textrm{regular}}|
\end{equation}
where $|\cdot|$ takes the absolute value, the subscript ``regular''
indicates the high precision solution computed by the RK4 method for
a sufficiently small time step $\Delta t=2^{-16}$ ms and ``library''
indicates the solution from the library method.

The error from the way of interpolation can be well described by the
error of the reset value $\Delta V^{\textrm{re}}$. We use a constant
input $I^{\textrm{input}}$ to drive the HH neuron to exclude the
influence of the assumption of constant input. Denote the sample intervals
$\Delta I,\Delta m,\Delta h,\Delta n$ for $I^{\textrm{th}},m^{\textrm{th}},h^{\textrm{th}},n^{\textrm{th}}$,
respectively. A linear interpolation yields an error of $\Delta V^{\textrm{re}}=O(\Delta I^{2}+\Delta m^{2}+\Delta h^{2}+\Delta n^{2})$.
In our simulation, we take $\Delta I=2.5$ $\mu\textrm{A\ensuremath{\cdot}cm}^{-2}$,
$\Delta m=\Delta h=\Delta n=0.02$ to build the library. If we use
sample intervals $c\Delta I,c\Delta m,c\Delta h,c\Delta n$, then
it is straight forward that $\Delta V^{\textrm{re}}=O(c^{2})$ as
shown in Fig. \ref{fig:LibError_interpolation}(a) (Same results hold
for $m,h,n$). In the same way, a cubic interpolation yields $\Delta V^{\textrm{re}}=O(c^{4})$
as shown in Fig. \ref{fig:LibError_interpolation}(b). Thus the cubic
interpolation can indeed improve the accuracy when the network is
driven by a constant input as shown in Fig. \ref{fig:LibError_interpolation}(c).

\begin{figure}[H]
	\centering{}\includegraphics[width=1\textwidth]{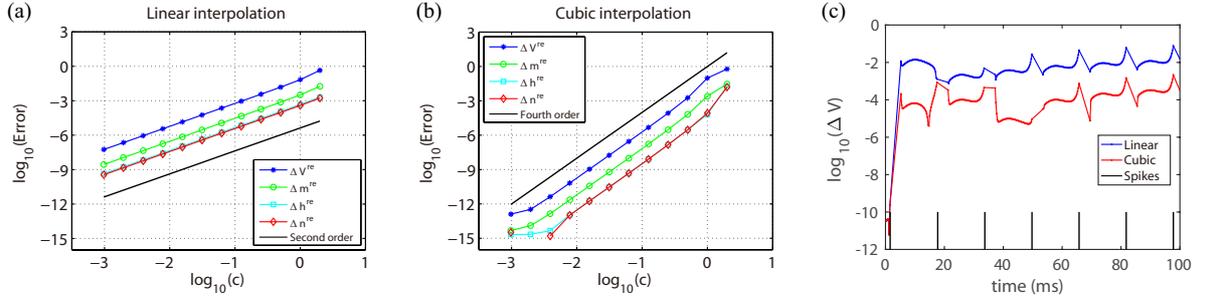}\caption{(a, b): The error of the reset values for a single HH neuron driven
		by constant input. (a) A linear interpolation and (b) a cubic interpolation
		to use the library. (c) The error of $\Delta V$ for a single HH neuron
		driven by constant input with a linear (blue) and cubic (red) interpolation
		to use the library. The short black lines indicate a spike. Time step
		is $\Delta t=2^{-16}$ ms for (a-c).\label{fig:LibError_interpolation}}
\end{figure}

We now consider the error due to the assumption of constant input.
Driven a single HH neuron by Poisson input, we compare the error of
$\Delta V$ with different time steps and ways of interpolation to
use the library as shown in Fig. \ref{fig:LibError_dV}. We find that
once the neuron fires a spike and calls the library, the reset error
$\Delta V^{\textrm{re}}$ (the dotted line in Fig. \ref{fig:LibError_dV})
is significantly greatly than the error of $\Delta V$ from the time
step and way of interpolation. Therefore, the biggest error is from
the assumption of constant input which cannot be reduced. Hence, we
build a relatively coarse library and choose the linear interpolation
to use the library in the Letter. Besides, we also find that the error
of $\Delta V$ will not accumulate linearly with call number. As shown
in Fig. \ref{fig:LibError_dV}, when the neuron fires, the error of
$\Delta V$ will first raise to the level of $\Delta V^{\textrm{re}}$
and then quickly decay until the next spike time. 

\begin{figure}[H]
	\begin{centering}
		\includegraphics[width=1\textwidth]{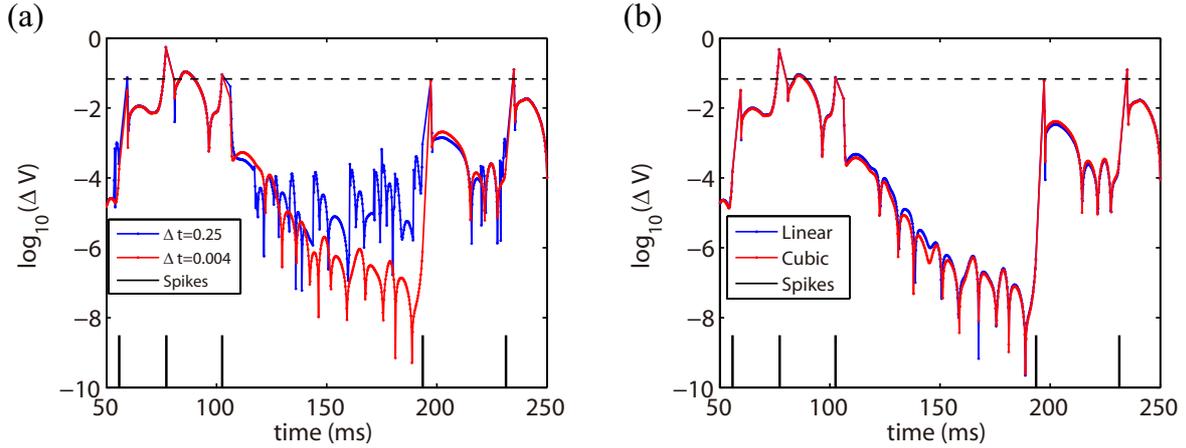}
		\par\end{centering}
	\caption{(a) The error of $\Delta V$ for one single neuron driven by Poisson
		input with a linear interpolation to use the library. Time steps are
		$\Delta t=0.25$ ms (blue) and $\Delta t=0.004$ ms (red). (b) The
		error of $\Delta V$ for one single neuron driven by Poisson input
		with a linear (blue) and cubic (red) interpolation to use the library,
		$\Delta t=0.004$ ms. The dotted line indicates the error of the reset
		value $\Delta V^{\textrm{re}}$ and the short black lines indicate
		a spike. \label{fig:LibError_dV}}
\end{figure}

\section{Numerical results \label{sec: Numerical result}}

\subsection{Hopf bifurcation of individual HH neuron \label{subsec:Hopf}}

In this section, we show the performance of the given adaptive, ETD4RK
and library methods with large time steps. We first exam whether the
given methods can capture the type II behavior of individual HH neuron.
Driven by constant input, the neuron can fire regularly and periodically
\cite{gerstner2002spiking,koch1998methods} only when the input current
greater than a critical value $I^{\textrm{input}}\approx6.2$ $\mu\textrm{A\ensuremath{\cdot}cm}^{-2}$.
The HH model has a sudden jump around this critical value from zero
firing rate to regular nonzero firing rate because of a subcritical
Hopf bifurcation \cite{koch1998methods}, as shown in Figure \ref{fig:Hopf}(a).
Below the critical value, some spikes may appear before the neuron
converges to stable zero firing rate state. The number of spikes during
this transient period depends on how close the constant input is to
the critical value, as shown in Figure \ref{fig:Hopf}(b). We point
out that all the adaptive, ETD4RK and library method with time steps
8 time greater than that of the regular method can capture the type
II behavior. Especially, the original library method in Ref. \cite{sun2009library}
fails to capture the spikes in the transient period since they use
stable information of $V,m,h,$ and $n$ to build the library. Our
library method uses the more intuitive suit of $V,m,h$ and $n$ to
build the library which already covers both the possible stable and
transient cases and is thus more precise.

\begin{figure}[h]
	\begin{centering}
		\includegraphics[width=1\textwidth]{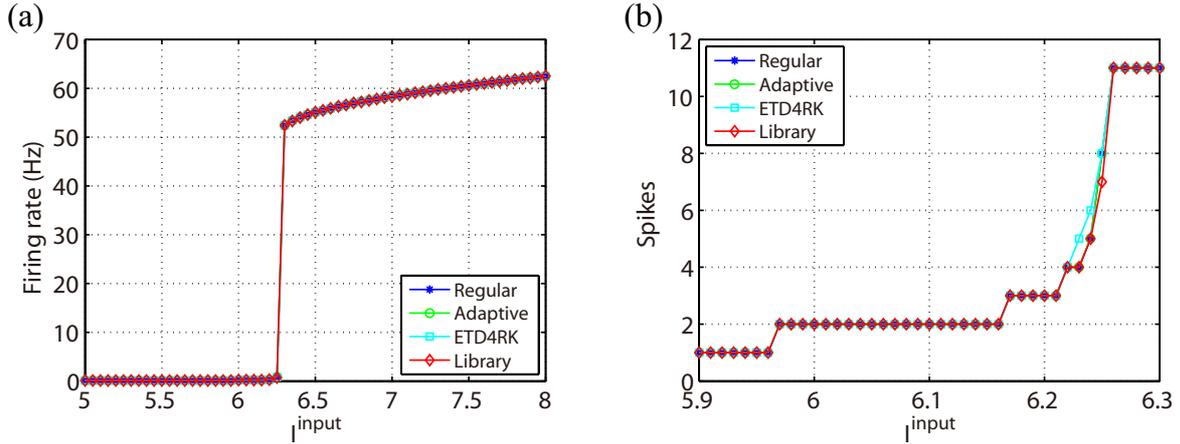}
		\par\end{centering}
	\caption{(a) The firing rate as a function of constant input $I^{\textrm{input}}$
		($\mu\textrm{A\ensuremath{\cdot}cm}^{-2}$). (b) The number of spikes
		during the transient period with initial voltage $V=-65$ mV. The
		blue stars, green circles, cyan squares and red diamonds in (a) and
		(b) indicate the results using regular method with $\Delta t=1/32$
		ms, adaptive, ETD4RK and library method with $\Delta t=1/4$ ms, respectively.
		\label{fig:Hopf}}
\end{figure}

\subsection{Lyapunov exponent\label{subsec:Lyapunov-exponent}}

For the performance of a network, we mainly consider a homogeneously
and randomly connected network of 100 excitatory neurons with connection
probability 10\%, feedforward Poisson input $\nu=100$ Hz and $f=0.1$
$\textrm{mS\ensuremath{\cdot}cm}^{-2}$. Then the coupling strength
$S$ is the only remaining variable. Other types of HH network and
other dynamic regimes can be easily extended and similar results can
be obtained. 

We first study the chaotic dynamical property of the HH system by
computing the largest Lyapunov exponent which is one of the most important
tools to characterize chaotic dynamics \cite{oseledec1968multiplicative}.
The spectrum of Lyapunov exponents can measure the average rate of
divergence or convergence of the reference and the initially perturbed
orbits \cite{ott2002chaos,thompson2002nonlinear,parker2012practical}.
If the largest Lyapunov exponent is positive, then the reference and
perturbed orbits will exponentially diverge and the dynamics is chaotic,
otherwise, they will exponentially converge and the dynamics is non-chaotic. 

When calculating the largest Lyapunov exponent, we use $\mathbf{X}=[X_{1},X_{2},...,X_{N}]$
to represent all the variables of the neurons in the HH model. Denote
the reference and perturbed trajectories by $\mathbf{X}(t)$ and $\mathbf{\tilde{X}}(t)$,
respectively. The largest Lyapunov exponent can be computed by

\begin{equation}
\lim_{t\rightarrow\infty}\lim_{\epsilon\rightarrow0}\frac{1}{t}\ln\frac{||\mathbf{\tilde{X}}(t)-\mathbf{X}(t)||}{||\epsilon||}\label{eq:LE}
\end{equation}
where $\epsilon$ is the initial separation. However we cannot use
Eq. (\ref{eq:LE}) to compute directly, because for a chaotic system
the separation $||\mathbf{\tilde{X}}(t)-\mathbf{X}(t)||$ is unbounded
as $t\rightarrow\infty$ and a numerical ill-condition will happen.
The standard algorithm to compute the largest Lyapunov exponent can
be found in Ref. \cite{parker2012practical,zhou2010spectrum,wolf1985determining}.
The regular, ETD4RK and adaptive methods can use these algorithms
directly. However, for the library method, the information of $V,m,h,n$
are blank during the stiff period and these algorithms do not work.
We use the extended algorithm introduced in Ref. \cite{zhou2009network}
to solve this problem.

As shown in Fig. \ref{fig:MLE}(a), we compute the largest Lyapunov
exponent as a function of coupling strength $S$ from 0 to 0.1 $\textrm{mS\ensuremath{\cdot}cm}^{-2}$
by the regular, adaptive, ETD4RK and library methods, respectively.
The total run time $T$ is 60 seconds which is sufficiently long to
have convergent results. The adaptive, ETD4RK and library methods
with large time steps ($\Delta t=1/4$ ms) can all obtain accurate
largest Lyapunov exponent compared with the regular method. It shows
three typical dynamical regimes that the system is chaotic in $0.037\lesssim S\lesssim0.068$
$\textrm{mS\ensuremath{\cdot}cm}^{-2}$ with positive largest Lyapunov
exponent and the system is non-chaotic in $0\lesssim S\lesssim0.037$
and $0.068\lesssim S\lesssim0.1$ $\textrm{mS\ensuremath{\cdot}cm}^{-2}$. 

As shown in Fig. \ref{fig:MLE}(b), we compute the mean firing rate,
denoted by $R$, obtained by these methods to further demonstrate
how accurate the adaptive, ETD4RK and library methods are. We give
the relative error in the mean firing rate, which is defined as

\begin{equation}
E_{\textrm{R}}=|R_{*}-R_{\textrm{regular}}|/R_{\textrm{regular}}
\end{equation}
where $*=\textrm{adaptive, ETD4RK, library}$. As shown in Fig. \ref{fig:MLE}(c),
all the given three methods can achieve at least 2 digits of accuracy
using large time steps ($\Delta t=1/4$ ms). Note that the relative
error may be zero in some cases and we set the logarithmic relative
error -8 in Fig. \ref{fig:MLE}(c), if it happens. Therefore, the
adaptive method has a bit advantage over the ETD4RK and library methods
for the mean firing rate. 

\begin{figure}[H]
	\begin{centering}
		\includegraphics[width=1\textwidth]{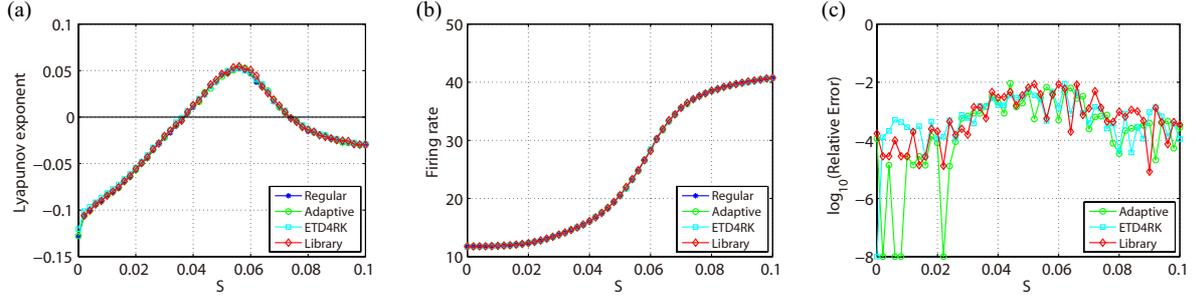}
		\par\end{centering}
	\caption{(a) The largest Lyapunov exponent of the HH network versus the coupling
		strength $S$. (b) Mean firing rate versus the coupling strength $S$.
		(c) The relative error in the mean firing rate (the Y-coordinate value
		-8 indicates there is no error). The blue, green, cyan and red curves
		represent the regular, adaptive, ETD4RK and library methods, respectively.
		The time steps are used as $\Delta t=1/32$ ms in the regular method
		and $\Delta t=1/4$ ms in the adaptive, ETD4RK and library methods.
		The total run time is 60 seconds to obtain convergent results. \label{fig:MLE}}
\end{figure}

From the calculation of the largest Lyapunov exponent, we have known
that there are three typical dynamical regimes in the HH model. As
shown in Fig. \ref{fig:Raster_258}, these three regimes are asynchronous
regime in $0\lesssim S\lesssim0.037$ $\textrm{mS\ensuremath{\cdot}cm}^{-2}$,
chaotic regime in $0.037\lesssim S\lesssim0.068$ $\textrm{mS\ensuremath{\cdot}cm}^{-2}$
and synchronous regime in $0.068\lesssim S\lesssim0.1$ $\textrm{mS\ensuremath{\cdot}cm}^{-2}$.
So we choose three coupling strength $S=0.02,0.05$ and $0.08$ $\textrm{mS\ensuremath{\cdot}cm}^{-2}$
to represent these three typical dynamical regimes respectively in
the following numerical tests.

\begin{figure}[H]
	\centering{}\includegraphics[width=1\textwidth]{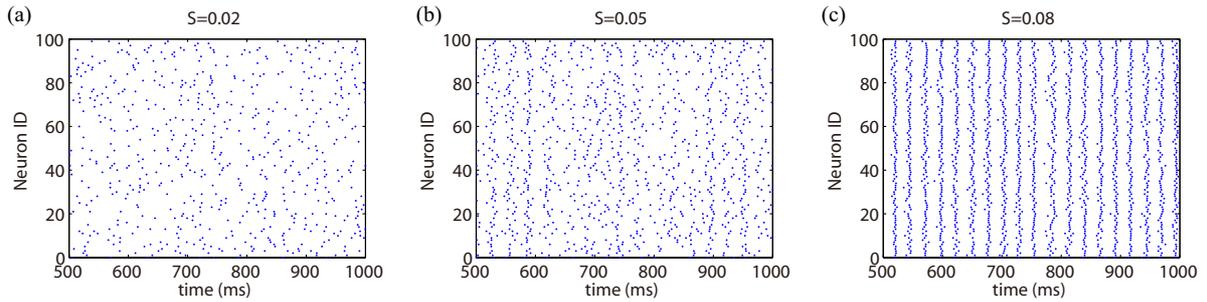}
	\caption{Raster plots of firing events in three typical dynamical regimes with
		coupling strength (a) $S=0.02$ $\textrm{mS\ensuremath{\cdot}cm}^{-2}$,
		(b) $S=0.05$ $\textrm{mS\ensuremath{\cdot}cm}^{-2}$, (c) $S=0.08$
		$\textrm{mS\ensuremath{\cdot}cm}^{-2}$. Since all the given methods
		have almost the same raster, we only show the result obtained by the
		regular method. \label{fig:Raster_258}}
\end{figure}

\subsection{Convergence tests}

We now verify whether the algorithms given above have a fourth-order
accuracy by performing convergence tests. For each method, we use
a sufficiently small time step $\Delta t=2^{-16}$ ms ($\Delta t_{S}=2^{-16}$
ms for adaptive method) to evolve the HH model to obtain a high precision
solution $\mathbf{X}^{\textrm{high}}$ at time $t=2000$ ms. To perform
a convergence test, we also compute the solution $\mathbf{X}^{(\Delta t)}$
using time steps $\Delta t=2^{-4},2^{-5},...,2^{-10}$ ms. The numerical
error is measured in the $L^{2}$ -norm

\begin{equation}
E=||\mathbf{X}^{(\Delta t)}-\mathbf{X}^{\textrm{high}}||
\end{equation}
As shown in Fig. \ref{fig:convergence test}(a), the regular method
can achieve a fourth-order accuracy for the non-chaotic regimes $S=0.02$
and $0.08$ $\textrm{mS\ensuremath{\cdot}cm}^{-2}$. For the chaotic
one $S=0.05$ $\textrm{mS\ensuremath{\cdot}cm}^{-2}$, we can never
achieve convergence of the numerical solution. The adaptive, ETD4RK
and library methods have similar convergence phenomena as shown in
Fig. \ref{fig:convergence test}. 

\begin{figure}[H]
	\begin{centering}
		\includegraphics[width=1\textwidth]{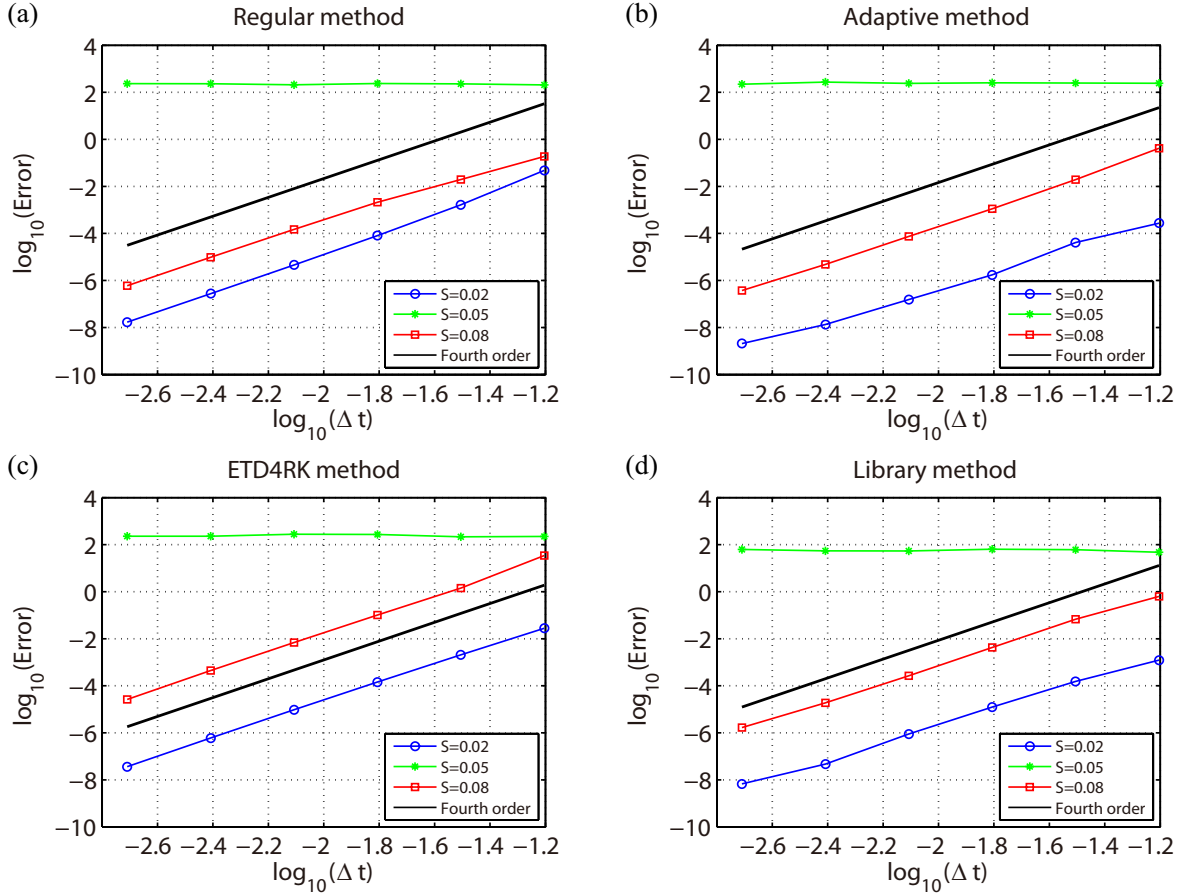}
		\par\end{centering}
	\caption{Convergence tests of (a) regular method, (b) adaptive method, (c)
		ETD4RK method and (d) library method. The convergence tests are performed
		by evolving the HH model for a total run time of $T=2000$ ms. We
		show the results for the coupling strength $S=0.02\textrm{ (blue)},0.05\textrm{ (green)}$
		and $0.08$ (red) $\textrm{mS\ensuremath{\cdot}cm}^{-2}$, respectively.
		\label{fig:convergence test}}
\end{figure}

\subsection{Computational efficiency }

In this section, we compare the efficiency among the adaptive, ETD4RK
and library methods. A straight forward way is to compare the time
cost of each method with a same total run time $T$. We notice that
all the methods are based on the standard RK4 scheme, so we can compare
the call number of the standard RK4 scheme that each method costs
to give an analytical comparison of efficiency. For the ETD4RK method,
the computational cost of one standard ETD4RK scheme is a bit more
than that of the strand RK4 scheme. Here we omit the difference and
no longer distinguish them for simplicity.

We start by exploring the call number of standard RK4 scheme for one
neuron for one single time step with the regular method, which is
presented in Theorem \ref{thm: call num RK4}. 
\begin{thm} \label{thm: call num RK4}
	Suppose the presynaptic spike train to
	a neuron can be modeled by a Poisson train with rate $\lambda$, then
	the call number of standard RK4 scheme for the neuron for one time
	step from $t_{0}$ to $t_{0}+\Delta t$ in the regular method is $1+\nu\Delta t+(2\Delta t+\nu\Delta t^{2})\lambda$
	on average.
\end{thm}

\begin{proof}
	During the initial preliminary evolving from $t_{0}$ to $t_{0}+\Delta t$,
	it requires a call number of $1+\nu\Delta t$ on average. Suppose
	there are $n$ spikes fired by the presynaptic neurons during $[t_{0},t_{0}+\Delta t]$
	which happens with probability $\frac{(\lambda\Delta t)^{n}}{n!}e^{-\lambda\Delta t}$.
	Denote the spike times by $t_{1},t_{2},...,t_{n}$, where $t_{0}<t_{1}<t_{2}<...,t_{n}<t_{0}+\Delta t$.
	For the $k$-th spike, as stated in Algorithm 1, it should first update
	from time $t_{k-1}$ to $t_{k}$ and then preliminarily evolve from
	$t_{k}$ to $t_{0}+\Delta t$ to update its next spike time $t_{\text{spike}}$.
	So the call number due to the $k$-th spike is $2+\nu(t_{0}+\Delta t-t_{k-1})=2+\nu\Delta t(1-\frac{t_{k-1}-t_{0}}{\Delta t})$
	on average. Then the total average extra call number due to presynaptic
	neurons is 
	\begin{equation}
	2n+\nu\Delta t+\nu\Delta t\sum_{k=1}^{n-1}(1-\frac{t_{k}-t_{0}}{\Delta t})
	\end{equation}
	Under the condition of $n$ spikes, the distribution of $\frac{t_{1}-t_{0}}{\Delta t},\frac{t_{2}-t_{0}}{\Delta t},...,\frac{t_{n}-t_{0}}{\Delta t}$
	is the same as the order statistics of $n$ independent and uniform
	distribution $U(0,1)$. Therefore, we have
	
	\begin{equation}
	\begin{aligned}\mathbb{E}\left(\nu\Delta t\sum_{k=1}^{n-1}(1-\frac{t_{k}-t_{0}}{\Delta t})\Big|n\textrm{ spikes}\right) & =\nu\Delta t\sum_{k=1}^{n-1}\intop_{0}^{1}\frac{n!}{(k-1)!(n-k)!}x^{k-1}(1-x)^{n-k+1}dx\\
	& =\nu\Delta t\frac{(n-1)(n+2)}{2(n+1)}
	\end{aligned}
	\end{equation}
	where $\mathbb{E}$ takes the expectation. Hence the expected extra
	call number due to these spikes is
	
	\begin{equation}
	\begin{aligned}\mathbb{E}\left(2n+\nu\Delta t+\nu\Delta t\frac{(n-1)(n+2)}{2(n+1)}\right) & =\sum_{n=1}^{\infty}\left(2n+\nu\Delta t+\nu\Delta t\frac{(n-1)(n+2)}{2(n+1)}\right)\frac{(\lambda\Delta t)^{n}}{n!}e^{-\lambda\Delta t}\\
	& =(2\lambda+\nu)\Delta t+\frac{\lambda\nu}{2}\Delta t^{2}+\frac{\nu}{\;\lambda}(e^{-\lambda\Delta t}-1)\\
	& =2\lambda\Delta t+\lambda\nu\Delta t^{2}+o(\Delta t^{2})
	\end{aligned}
	\end{equation}
	Therefore, the expected call number for one neuron for one time step
	in regular method is $1+\nu\Delta t+(2\Delta t+\nu\Delta t^{2})\lambda$.
\end{proof}
We suppose the spike events of presynaptic neurons are a Poisson train
since they can be asymptotically approach a Poisson process \cite{lewis1972stochastic,cinlar1968superposition}
when the spike timing are irregular. Once a neuron fires a spike,
it should also update to this spike time and then preliminarily evolve
to update its next spike time $t_{\textrm{spike}}$, so the Poisson
spike train it goes through has rate $(1+p_{c}N)R$, where the 1 corresponds
to the spikes it fires and $p_{c}$ is the connection probability.
Therefore, the average call number per neuron for the regular method
is 

\begin{equation}
\left\langle \mathcal{N}_{\textrm{regular}}\right\rangle =\frac{T}{\Delta t}[1+\nu\Delta t+(2\Delta t+\nu\Delta t^{2})(1+p_{c}N)R]\label{eq:call re}
\end{equation}
Note that regular and ETD4RK method have the same formula of call
number. 

For the adaptive method, it is based on the regular method, so it
also has the required call number in Eq. (\ref{eq:call re}) . Besides,
we should add the extra call number due to the smaller time step $\Delta t_{S}$
during the stiff period. For the spike the neuron fires, it requires
extra call number of $TRT^{\textrm{stiff}}/\Delta t_{S}$. For the
presynaptic spikes, similar to Theorem \ref{thm: call num RK4}, the
extra call number is $p_{c}NTRp\Delta t/\Delta t_{S}$, where $p$
is the probability that the neuron is in the stiff period when its
presynaptic neurons fire. Therefore, we have

\begin{equation}
\begin{aligned}\left\langle \mathcal{N}_{\textrm{adaptive}}\right\rangle  & \approx\frac{T}{\Delta t}[1+\nu\Delta t+(2\Delta t+\nu\Delta t^{2})(1+p_{c}N)R]+TR\frac{T^{\textrm{stiff}}}{\Delta t_{S}}+p_{c}NTRp\frac{\Delta t}{\Delta t_{S}}\end{aligned}
\label{eq:call ad}
\end{equation}
For the library method, once a neuron fires a spike, we do not evolve
its $V,m,h,n$ for the next $T^{\textrm{stiff}}$ ms, $\mathit{i.e.}$,
no call of the standard RK4 scheme. With a little amendment of Eq.
(\ref{eq:call re}), we have 

\begin{equation}
\left\langle \mathcal{N}_{\textrm{library}}\right\rangle \approx(1-T^{\textrm{stiff}}R)\frac{T}{\Delta t}(1+\nu\Delta t)+\frac{T}{\Delta t}(2\Delta t+\nu\Delta t^{2})(R+p_{c}NR(1-p))\label{eq:call lib}
\end{equation}
where $T^{\textrm{stiff}}R$ is the probability that a neuron stays
in the stiff period.

\begin{figure}[H]
	\begin{centering}
		\includegraphics[width=1\textwidth]{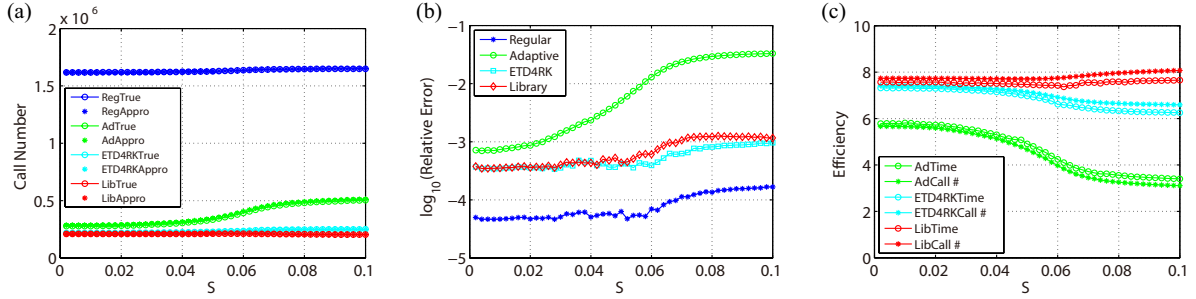}
		\par\end{centering}
	\caption{(a) The call number of standard RK4 scheme per neuron for the regular,
		adaptive, ETD4RK and library methods. The circles represent the call
		number from numerical tests while the stars represent the approximations
		from Eqs. (\ref{eq:call re})-(\ref{eq:call lib}). (b) The relative
		error in the call number between the approximation and numerical tests
		for the regular, adaptive, ETD4RK and library methods. (c) The efficiency
		measured by the time cost and call number versus the coupling strength
		$S$ for the adaptive, ETD4RK and library methods. Time steps and
		colors are the same as the one in Fig. \ref{fig:MLE}. Total run time
		is 50 seconds. \label{fig: RK4 call numbers}}
\end{figure}

As shown in Fig. \ref{fig: RK4 call numbers}, we count the call numbers
from the numerical tests and approximations in Eqs. (\ref{eq:call re}-\ref{eq:call lib})
for the regular, adaptive, ETD4RK and library methods, respectively.
These equations are indeed close approximations of the call number
achieving at least 1 digit of accuracy. We should point out that the
information of $R$ and $p$ in the approximations are obtained from
the numerical tests since they are hard to estimate directly. 

Fig. \ref{fig: RK4 call numbers}(c) gives the efficiency of the adaptive,
ETD4RK and library methods by comparing the time cost and call number
with that of the regular method. These two kinds of efficiency are
quite consistent except a bit difference in the ETD4RK and library
methods. For the ETD4RK method, this is because we use the standard
ETD4RK scheme during the stiff period which costs a bit more time
than the standard RK4 scheme. So the efficiency measured by the call
number is a bit overestimated. For the library method, this is because
when a neuron fires a spike, we should call the library and evolve
the parameters $G$ and $H$ during the stiff period which costs some
time but is not included in the efficiency measured by the call number.
When the mean firing rate is high, this extra consumed time is no
longer negligible. Even so, these two kinds of efficiency still show
good agreement for the ETD4RK and library methods. Therefore, we can
use the efficiency measured by the call number to compare.

\begin{figure}[H]
	\centering{}\includegraphics[width=1\textwidth]{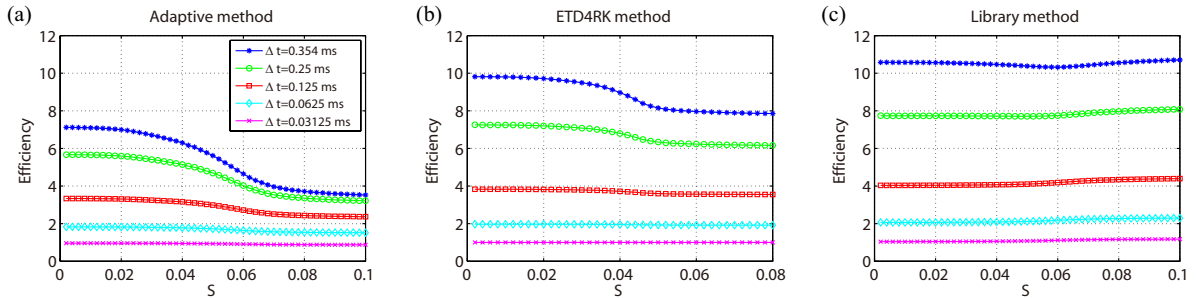}
	\caption{Efficiency measured by the expected call number for the adaptive,
		ETD4RK and library methods, respectively. The time steps are $\Delta t=0.354,0.25,0.125,0.0625,0.03125$
		ms from top to bottom. Total run time is 50 seconds. \label{fig:efficiency}}
\end{figure}

Fig. \ref{fig:efficiency} shows the efficiency for the adaptive,
ETD4RK and library methods with different time steps. When the mean
firing rate is low ($\leq20$ Hz), the adaptive method can achieve
high efficiency with a maximum 7 times of speedup. However this high
efficiency decreases rapidly with the coupling strength (mean firing
rate), since the neurons will use the smaller time step much more
often when the firing rate is high. The ETD4RK and library methods
are not so sensitive to the firing rate and can achieve a global high
efficiency. Especially, the library method can achieve a maximum 10
times of speedup for the maximum time step $\Delta t=0.354$ ms.

\section{Conclusion}

In conclusion, we have shown three methods to deal with the stiff
part during the firing event in evolving the HH equations. All of
them can use a large time step to save computational cost. The adaptive
method is the easiest to use and can retain almost all the information
of the original HH neurons like the spike shapes. However, it has
limited efficiency and is more appropriate for the dynamical regimes
with low firing rate. The ETD4RK method can obtain both precise trajectory
of the variables and high efficiency. The library method achieves
the highest efficiency, but it sacrifices the accuracy of spike shapes
and is more suitable for high accurate statistical information of
HH neuronal networks.

For the adaptive method, we should point out that the standard adaptive
method like the Runge-Kutta-Fehlberg method is not suitable for the
HH neuronal networks. When evolving a single neuron, the standard
adaptive method requires an extra RK4 calculation to decide the adaptive
time step. When the neuron fires, its equations become stiff and the
chosen time step is very small $\sim1/32$ ms as illustrated in Ref.
\cite{borgers2013exponential}. So it is no better than our adaptive
method for a single neuron. For networks, the standard adaptive method
should evolve the entire network to choose the adaptive time step.
For large-scaled networks, there are firing events almost everywhere,
then the chosen adaptive time step is always quite small $1/32$ ms
\cite{borgers2013exponential} and makes the method ineffective.

Our library method is based on the original work in Ref. \cite{sun2009library}.
The main difference is the way to build and use the library. In the
original work, once the membrane potential reaches the threshold,
we should record the input current $I^{\textrm{th}}$ and gating variables
$m^{\textrm{th}},h^{\textrm{th}},n^{\textrm{th}}$. Driving a single
neuron with constant input $I^{\textrm{th}}$, we will obtain a periodic
trajectory of membrane potential after the transient period. Especially
we intercept a stable section of the trajectory whose initial value
is the threshold $V^{\text{th}}$ and data length is the stiff period.
Given the section of membrane potential and initial values $m^{\textrm{th}},h^{\textrm{th}},n^{\textrm{th}}$,
we can obtain the corresponding reset gating variables $m^{\textrm{re}},h^{\textrm{re}},n^{\textrm{re}}$,
respectively. Denote the sample number of $I^{\textrm{th}},m^{\textrm{th}},h^{\textrm{th}},n^{\textrm{th}}$
by $N_{I},N_{m},N_{n},N_{h}$, then the size of library is $N_{I}(1+N_{m}+N_{n}+N_{h})$,
$i.e.$, the input current $I^{\textrm{th}}$ is the most important.
In our library method, the gating variables are as important as the
input current. Given a suite of $I^{\textrm{th}},m^{\textrm{th}},h^{\textrm{th}},n^{\textrm{th}}$,
we evolve the HH equations with constant input $I^{\textrm{th}}$
for $T^{\textrm{stiff}}$, then we can obtain the corresponding reset
values $I^{\textrm{re}},m^{\textrm{re}},h^{\textrm{re}},n^{\textrm{re}}$
simultaneously. The advantage of our library method lies in two aspects:
1) It is much easier to build the library. 2) Our data library is
much accurate since it can cover both the transient and the stable
periodic information. 

Finally, we emphasize that the spike-spike correction procedure \cite{rangan2007fast}
is necessary in the given methods to achieve an accuracy of fourth-order.
However, it will greatly decrease the efficiency of the advanced ETD4RK
and library methods for large-scaled networks. When there are many
neurons that fire in one time step, each neuron will call the standard
RK4 scheme a lot during the updating and preliminary evolving procedure.
When the size of the network tends infinity, the given methods with
large time step will even slower than the regular method with a small
time step as illustrated in Eqs. (\ref{eq:call re}), (\ref{eq:call ad})
and \ref{eq:call lib}). We point out that this problem can be solved
by reducing the accuracy from the fourth-order to a second-order.
In each time step, we evolve each neuron without considering the feedforward
and synaptic spikes and recalibrate their effects at the end of the
time step. Then the spike-spike correction procedure is avoided while
the methods still have an accuracy of second-order \cite{shelley2001efficient}.
Besides, the call number per neuron is merely $T/\Delta t$ for regular
and ETD4RK methods and $(1-T^{\textrm{stiff}}R)T/\Delta t$ for the
library method. Therefore, the ETD4RK and library methods can stably
achieve over 10 times of speedup with a large time step in any kinds
of networks, $e.g.$, all-to-all connected networks, large-scaled networks
and network of both excitatory and inhibitory neurons.



%
%
%


\appendix
\section*{Appendix: Parameters and variables for the Hodgkin-Huxley equations}
\label{sec:Appendix}
Definitions of $\alpha$ and $\beta$ are as follows \cite{dayan2001theoretical}: 

$\alpha_{m}(V)=(0.1V+4)/(1-\exp(-0.1V-4)),$

$\beta_{m}(V)=4\exp(-(V+65)/18),$

$\alpha_{h}(V)=0.07\exp(-(V+65)/20),$

$\beta_{h}(V)=1/(1+\exp(-3.5-0.1V)),$

$\alpha_{n}(V)=(0.01V+0.55)/(1-\exp(-0.1V-5.5)),$

$\beta_{n}(V)=0.125\exp(-(V+65)/80)$.

Other model parameters are set as $C=1\mu\textrm{F\ensuremath{\cdot}cm}^{-2}$,
$V_{Na}=50$ mV, $V_{K}=-77$ mV, $V_{L}=-54.387$ mV, $G_{Na}=120\textrm{ mS\ensuremath{\cdot}cm}^{-2}$,
$G_{K}=36\textrm{ mS\ensuremath{\cdot}cm}^{-2}$, $G_{L}=0.3\textrm{ mS\ensuremath{\cdot}cm}^{-2}$,
$V_{G}=0$ mV, $\sigma_{r}=0.5$ ms, and $\sigma_{d}=3.0$ ms \cite{dayan2001theoretical}.

\bibliographystyle{unsrt}
\bibliography{reference}


\end{document}